\documentclass{fundam}

\usepackage{amssymb} 
\usepackage[bookmarks=false]{hyperref}
\usepackage{mymacros}

\begin{document}
\title{Leakage-resilient Algebraic Manipulation Detection Codes with Optimal Parameters} 

\author{Divesh~Aggarwal\\
National University of Singapore
\and Tomasz~Kazana\\
University of Warsaw
\and Maciej~Obremski\\
National University of Singapore}

\runninghead{D. Aggarwal, T. Kazana, M. Obremski}{Leakage-resilient Algebraic Manipulation Detection Codes}

\maketitle

\begin{abstract}
Algebraic Manipulation Detection (AMD) codes
is a cryptographic primitive that was introduced by Cramer, Dodis, Fehr, Padr{\'o} and Wichs~\cite{CDFPW08}. They are keyless message authentication codes that protect messages against additive tampering by the adversary assuming that the adversary cannot ``see" the codeword. For certain applications, it is unreasonable to assume that the adversary computes the added offset without any knowledge of the codeword $c$. Recently, Ahmadi and Safavi-Naini~\cite{AS13}, and then Lin, Safavi-Naini, and Wang~\cite{LSW16} gave a construction of leakage-resilient AMD codes where the adversary has some partial information about the codeword before choosing added offset, and the scheme is secure even conditioned on this partial information.

In this paper we establish bounds on the leakage rate $\rho$  and the code rate $\kappa$ 
for leakage-resilient AMD codes. In particular we prove that $2\rho+\kappa <1$ and for the weak case (security is averaged over a uniformly random message) $\rho+\kappa <1$. These bounds hold even if adversary is polynomial-time bounded, as long as we allow leakage function to be arbitrary.

We present constructions of AMD codes that (asymptotically) fulfill the above bounds for almost full range of parameters $\rho$ and $\kappa$. This shows that the above bounds and constructions are in-fact optimal. 

In the last section we show that if a leakage function is computationally bounded (we use the Ideal Cipher Model) then it is possible to break these bounds.
\end{abstract}

\section{Introduction}
\label{sec:intro}

Algebraic Manipulation Detection (AMD) codes~\cite{CDFPW08} are keyless message authentication codes that protect messages against additive tampering by the adversary assuming that the adversary cannot ``see" the codeword. In AMD codes, a message $m \in \bits^k$ is encoded to a codeword $C$ in $\bits^n$, and the codeword is stored such that the adversary cannot get any information about the codeword. The adversary is assumed to be able to add an arbitrary element $\mathsf{A}$ to $C$ such that $C + \mathsf{A}$ could potentially decode to a message $m' \neq m$. In a $\delta$-secure AMD code, such a manipulation succeeds with probability $\delta$, and with probability $1-\delta$, the decoder on input $C + \mathsf{A}$, either outputs $m$ or a special symbol $\bot$ indicating that the tampering (by the adversary) has been detected. Another notion that has been considered in~\cite{CDFPW08} is that of weakly secure AMD codes (also called weak AMD codes), where the security guarantee is only for a uniformly random message over the message space $\bits^k$, and the coding scheme is deterministic. 

As mentioned in~\cite{CDFPW08}, AMD codes find useful applications in linear secret sharing schemes (e.g. Shamir's secret sharing~\cite{Sha79}) and Fuzzy Extractors~\cite{DORS06}. In particular, AMD codes can be used to turn any linear secret sharing scheme into a so called {\em robust secret sharing scheme}~\cite{TW88}, which ensures that no unqualified subset of players can modify their shares and cause the reconstruction of a string $s'$ which is different from the secret $s$. Similarly, AMD codes can help turn fuzzy extractors into robust fuzzy extractors that were first considered by Boyen et al.~\cite{Boy04,BDKOS05}. We direct the reader to~\cite{CDFPW08} for a more detailed discussion of these applications. 

In reality, the assumption that the adversary computes the offset $\mathsf{A}$ without any knowledge of the codeword $c$ might be unrealistic due to the presence of side channel attacks~\cite{Koc96,KJJ99,AARR02}, where the adversary can obtain partial information (leakage) on the secret state of an implementation of a cryptographic scheme, by exploiting physical phenomena.  

Recently, Ahmadi and Safavi-Naini~\cite{AS13}, and then Lin, Safavi-Naini, and Wang~\cite{LSW16} gave a construction of so called $\rho$-Linear Leakage-Resilient AMD ($\rho$-LLR-AMD) codes where the adversary has some partial information about the codeword $c$ before choosing $\mathsf{A}$, and the scheme is secure even conditioned on this partial information. In~\cite{AS13}, the authors consider the notion of a coding scheme from $m \in \bits^k$ to $c \in \bits^n$ where the encoding algorithm uses randomness $R \in \bits^\sigma$, and the adversary computes $\mathsf{A}$ given partial information $Z$ such that the entropy of $R$ conditioned on $Z$ is at least $(1-\rho) \sigma$. A similar notion of weak $\rho$-LLR-AMD codes was defined and constructed where the security is for a uniformly random message $M$, and the entropy of $M$ conditioned on $Z$ is assumed to be at least $(1-\rho)k$. 

In the subsequent work, Lin, Safavi-Naini, and Wang~\cite{LSW16} considered a stronger notion of $\rho$-AMD codes, where $Z$ carries information about the codeword, and the entropy of the codeword $C$ conditioned on $Z$ is at least $(1-\rho) n$. Similar to the original AMD codes defined in~\cite{CDFPW08}, the authors defined weak and strong $\rho$-AMD codes as deterministic and randomized codes that guarantee security for a uniformly distributed message and any message, respectively. Since $\rho$-AMD codes are the main topic of our paper, we briefly restate the main application of $\rho$-AMD codes as discussed in~\cite{LSW16}. 

\paragraph*{Robust ramp secret sharing scheme.} A $(t,r)$-ramp secret sharing scheme~\cite{BM84,IY06} is a secret sharing scheme such that any $t$ or fewer shares reveal nothing about the secret $s$, and any $r$ or more shares are enough to reconstruct the secret. If the number of shares $a$ is between $r$ and $t$, then an $\frac{a-t}{r-t}$ fraction of the secret is leaked. Some authors such as~\cite{JM96} refer to this as a linear ramp secret sharing scheme. 
 By encoding a secret with a $\rho$-AMD code with error $\delta$, and then using a $(t,r)$-ramp secret sharing scheme, we can ensure that as long as the number of shares are at most $t + \lfloor \rho(r - t) \rfloor$, the probability of being able to reconstruct the secret is upper bounded by $\delta$. Notice that if we assume that the secret is chosen uniformly at random, then even a weak $\rho$-AMD code will be sufficient for this application. 

For this application, or for that matter any other application of $\rho$-AMD codes, we want the leakage fraction $\rho$ to be as large as possible and for the efficiency of the scheme, we additionally want the rate of the codeword $\kappa := \frac{k}{n}$ also to be as large as possible. In~\cite{LSW16}, the authors give a construction of strong $\rho$-AMD codes with error $\delta$, where $\kappa = \frac{d}{d+2}$, and $\rho = \frac{1}{d+2} - \eps$, where $\eps$ is a small constant that depends on $\delta$, and $d$ is a positive integer. In order to maximise the leakage, we can set $d = 1$, which will imply that $\rho \approx \frac{1}{3}$, and the rate of the code is $\frac{1}{3}$. Also, it was shown in~\cite{LSW16} that for any strong $\rho$-AMD code with any error $\delta$, we must have that $\kappa + \rho < 1$. This leads us to the following question.

\begin{question}
\label{que:strong}
 Does there exist a strong $\rho$-AMD code with leakage rate $\rho \ge \frac{1}{3}$? Can we obtain a better tradeoff between $\kappa$ and $\rho$?
\end{question}

In this paper, we answer both these questions in the affirmative. In Section~\ref{sec:const}, we generalise the construction from~\cite{LSW16} to obtain a construction of a whole family of $\rho$-AMD codes for a wider range of parameters. More precisely (see Corollary \ref{cor:strong-AMD} for details), we have constructions that are secure as long as $2 \rho + \kappa < 1$. Moreover, we show in Section~\ref{sec:lower-bound} that there exists no construction of strong $\rho$-AMD codes that is secure if $2\rho + \kappa \ge 1$. This means that we covered the whole space of possible values of $\rho$ and $\kappa$. On the other hand, we prove that if we work in Ideal Cipher Model we can go even further: we can break proven barrier and achieve $\rho$ arbitrary close to $1$. See Section~\ref{breakbound} for details.

Similarly, as above, for weak $\rho$-AMD codes with error $\delta$, Lin et al. gave a construction with $\kappa = \frac{d}{d+1}$, and $\rho = \frac{1}{d+1} - \eps$, where $\eps$ is a small constant that depends on $\delta$, and $d$ is a positive integer. Setting $d = 1$, we get $\rho \approx \frac{1}{2}$, and $\kappa = \frac{1}{2}$. They, however, failed to obtain any nontrivial condition under which there exist weak $\rho$-AMD codes. We can again ask a question similar to Question~\ref{que:strong} for weak $\rho$-AMD codes. 

\begin{question}
 Does there exist a weak $\rho$-AMD code with leakage rate $\rho \ge \frac{1}{2}$? Can we obtain some nontrivial tradeoff between $\kappa$ and $\rho$?
\end{question}

We answer the first question in the negative and the second in the affirmative by showing in Section~\ref{sec:lower-bound} that for any $\rho \in (0,1)$ there exists no weak $\rho$-AMD code with $\rho + \kappa \ge 1$, or $\rho \ge \kappa$. In other words, for any secure weak $\rho$-AMD code, we must have $\rho + \kappa < 1$, and $\rho < \kappa$. We also include a construction achieving parameters similar to~\cite{LSW16} in Section~\ref{sec:const}. 

We would again like to remark that all our constructions and proofs in Section~\ref{sec:const} closely resemble those in~\cite{CDFPW08,LSW16}. Our main contribution is to show that these constructions are optimal and that we can cover the whole space of  feasible parameters.

In the following, we briefly discuss the leakage models that have been considered in the literature.

\subsection{Leakage models}

In general, we consider the leakage as $f(x)$ where $f$ is the leakage function, and $x$ is the secret state. Notice that for achieving any security guarantee, we have to restrict the amount of leakage allowed, or else the adversary can leak the entire state. 
\paragraph{Bounded Leakage} The most widely used leakage model is the bounded leakage model introduced by Dziembowski and Pietrzak~\cite{DP08}. The popularity of this model is in part due to its simplicity. In this model, it is assumed that $|f(x)|$ is bounded, and (significantly) smaller than $|x|$. 

\paragraph{Noisy Leakage} The main limitation of the bounded leakage model is that real world side channel information obtained by the adversary is rarely bounded in length. Typically, the length is huge, but the correlation between the leakage and the secret state is small, i.e., most of the leakage is ``noise". This model was first considered by Naor and Segev~\cite{NS09} informally, and was more formally introduced by~\cite{DORS08}. In this model, the idea of weak correlation between the leakage and the secret is captured by the assumption that the average min-entropy of $x$ given $f(x)$ remains high. 

As shown in Lemma~\ref{lem:condmin}, the noisy leakage model is more general than the bounded leakage model. Thus, all results in this paper are in the noisy leakage model.

\section{Preliminaries}

For an integer $m \in \mathbb{N}$, we denote the
set of integers $\{1, \ldots, m\}$ by $[m]$. Unless otherwise stated, $\F = \F_q$ denotes a finite field of size $q$. 

The {\em min-entropy} of a
random variable $X$ is defined as $\hinf(X) \eqd - \log(\max_x
\Pr[X=x])$.
This notion is fundamental for cryptographic community as it it is used to measure the probability of guessing a random element drawn out from the distribution $X$.

We also define \emph{average (aka conditional)
min-entropy} of a random variable $X$ conditioned on another random
variable $Z$ as
\begin{eqnarray*}
\thinf(X|Z) &\eqd& -\log\left(\Ex_{z \leftarrow Z} \left[~\max_x\Pr[X=x|Z=z]~\right]\right) \\
&=& -\log\left(\Ex_{z\leftarrow Z}\left[2^{-\hinf(X|Z=z)}\right]\right) \;,
\end{eqnarray*}

where $\Ex_{z\gets Z}$ denotes the expected value over $z\gets Z$. Intuitively, if $\thinf(X|Z)$ is high, then almost certainly (for all except only negligible fraction of events from $Z$) the probability of guessing $X$ conditioned on some $z\gets Z$ is low. Placing the operator $\Ex$ inside the logarithm is crucial in the definition to hold this intuition. For more discussion on this notion one is referred to \cite{DORS06}.

The following lemma shows a fundamental property of conditional min-entropy.
\begin{lemma}[\cite{DORS08}]\label{lem:condmin}
	Given distributions $X$, $Y$ where $|\mathsf{support}(Y)|\leq 2^\lambda$, we have that
	\begin{equation*}
		\thinf(X|Y)\geq \hinf(X,Y)-\lambda\geq \hinf(X)-\lambda.
	\end{equation*}
\end{lemma}

We will need also the following result:
\begin{lemma}
\label{lem:bonferroni}
Let $0<p<1$, and let $E_1, \ldots, E_t$ be pairwise independent events such that $\Pr(E_i) = p$ for all $i \in [t]$. Then
\[
\Pr(\cup_{i=1}^t E_i) \ge t \cdot p - \frac{t^2 \cdot p^2}{2} \;.
\]
\end{lemma}
\begin{proof}
Using Bonferroni inequality~\cite{Bon36}, we have that
\begin{eqnarray*}
\Pr(\cup_{i=1}^t E_i) &\ge& \sum_{i=1}^t \Pr(E_i) - \sum_{i, j \in [t], i \neq j} \Pr(E_i \cap E_j) \\
                                 &=& t \cdot p - \frac{t (t-1)}{2} p^2 \\
                                 &>& t \cdot p - \frac{t^2 \cdot p^2}{2} \;.
\end{eqnarray*}
\end{proof}

\section{Definitions}
We first define a general coding scheme.
\begin{definition} 
\label{coding}
A coding scheme is given by an encoding function $\Enc:\F^k \times \F^\sigma \mapsto \F^n$ from $k$-length messages to $n$-length codewords\footnote{The encoding function takes randomness of length $\sigma$, which we make explicit in the definition for convenience.}, and a decoding function $\Dec: \F^n \mapsto \F^k \cup \{ \bot \}$ such that, for each $m \in \F^k$, $r \in \F^\sigma$, we have that $\Pr(\Dec(\Enc(m,r)) = m) = 1$.

Additionally, the coding scheme is called regular if $\Enc$ is a one-to-one function. 
 \end{definition}

We now define AMD codes.
\begin{definition} 
\label{def:AMD}
 Let $\Enc:\F^k \times \F^\sigma \mapsto \F^n$, $\Dec: \F^n \mapsto \F^k \cup \{ \bot \}$ be a coding scheme. We say that ($\Enc$, $\Dec$) is a strong ($\rho, \delta$)-AMD code if for any $m \in \F^k$, $R$ uniform in $\F^\sigma$, \[\Pr[\Dec(\Enc(m,R) + \mathsf{A}(Z)) \notin \{m, \bot\} ] \le \delta \;,\] where $Z \in \cZ$ is a   random variable such that $\thinf(\Enc(m,R) | Z) \ge \hinf(\Enc(m,R)) - \rho \cdot( n \log q)$, and $\mathsf{A}: \cZ \mapsto \F^n$ is an arbitrary randomized function chosen by the adversary. The probability is over the randomness of the encoding, the possibly independent randomness of the leakage function, and the randomness of the adversary, i.e., over $R, Z, \mathsf{A}$. 
 
  If the adversary is only allowed time polynomial in $n$ to compute $\mathsf{A}(Z)$, then the underlying scheme is said to be a \emph{computationally secure strong} ($\rho,\delta$)-\emph{AMD code}.
 \end{definition}
 
If the security guarantee is only for a uniform message distribution, then we call such an AMD code a weak AMD code. More formally,

\begin{definition} 
\label{def:weak-AMD} 
 Let $\Enc:\F^k  \mapsto \F^n$, $\Dec: \F^n \mapsto \F^k \cup \{ \bot \}$ be a coding scheme.
 \footnote{In this paper, we restrict our attention to weak leakage resilient AMD codes that are deterministic as in~\cite{AS13,LSW16}. In general, when talking about AMD codes without leakage-resilience, we typically don't assume that the encoding function is deterministic. 
 }
 We say that ($\Enc$, $\Dec$) is a weak ($\rho, \delta$)-AMD code if for $M$ uniform in $\F^k$, we have that \[\Pr[\Dec(\Enc(M) + \mathsf{A}(Z)) \notin \{M, \bot\} ] \le \delta \;,\] where $Z \in \cZ$ is a leakage variable such that $\thinf(\Enc(M) | Z) \ge \hinf(\Enc(M)) - \rho n \log q$, and $A: \cZ \mapsto \F^n$ is an arbitrary function chosen by the adversary. 
 
 If the adversary is only allowed time polynomial in $n$ to compute $\mathsf{A}(Z)$, then the underlying scheme is said to be a computationally secure weak ($\rho,\delta$)-AMD code.
 \end{definition}

\section{Constructing Leakage-resilient AMD codes}
\label{sec:const}

In the following, we show that given AMD codes, we can construct leakage-resilient AMD codes. 
\begin{lemma}\label{lem:creatingleakage}
For any $\delta > 0$, $0 < \rho < 1$, any regular coding scheme $\Enc:\F^k \times \F^\sigma \mapsto \F^n$, $\Dec: \F^n \mapsto \F^k \cup \{ \bot \}$ that is a strong ($0, \delta$)-AMD code is also a strong ($\rho, q^{\rho n} \delta$)-AMD code.
\end{lemma}
\begin{proof}
Since $(\Enc, \Dec)$ is a $(0,\delta)$-AMD code, we have that for a uniform $R$ in $\F^\sigma$, and any $m \in \F^k, \alpha \in \F^n$,
\begin{equation*}
\Pr(\Dec(\Enc(m, R) + \alpha) \notin \{m, \bot\}) \le \delta \;.
\end{equation*}
Define $\BAD(m, \alpha)$ to be the set of all $c$ in the support of $\Enc(m,R)$ such that $\Dec(c + \alpha) \notin \{m, \bot\}$. From the equation above, we have that
\begin{equation}
\label{eq:no-leakage}
|\BAD(m,\alpha)| \le \delta \cdot q^\sigma \;.
\end{equation} 
Now, consider a leakage variable $Z$ such that $\thinf(\Enc(m,R) | Z) \ge \hinf(\Enc(m,R)) - \rho n \log q$. Since $(\Enc,\Dec)$ is a regular coding scheme, we have that $\hinf(\Enc(m,R)) = \hinf(R) = \sigma \log q$, and hence $\thinf(\Enc(m,R) | Z) \ge (\sigma - \rho n) \log q$. Thus, using the definition of conditional min-entropy, we have that
\begin{equation}
\label{eq:cond-min-entropy}
\sum_{z \in \cZ} \Pr(Z = z) \cdot \max_{c \in \F^n} \Pr(\Enc(m, R) = c \; | \; Z = z)  \le \frac{1}{q^{\sigma - \rho n}} \;.
\end{equation}
We now bound the probability of incorrect decoding when the adversary computes the offset given $Z$. 
\begin{align*}
& \Pr[\Dec(\Enc(m,R) + \mathsf{A}(Z)) \notin \{m, \bot\} ] \\
=&  \sum_{z \in \cZ}  \Big(\Pr[\Dec(\Enc(m,R) + \mathsf{A}(Z)) \notin \{m, \bot\} \; |\; Z=z ] 
 \cdot \Pr[Z=z] \Big) \\
=&  \sum_{z \in \cZ}  \Pr[\Enc(m,R) \in \BAD(m, \mathsf{A}(Z)) \; |\; Z=z ] \cdot \Pr[Z=z] \\
\le&  \sum_{z \in \cZ}  \Big( |\BAD(m, \mathsf{A}(Z)| \max_{c \in \F^n} \Pr(\Enc(m, R) = c \; | \; Z = z) 
 \cdot \Pr(Z = z) \Big) \\
\le&  \delta \cdot q^\sigma \cdot \frac{1}{q^{\sigma - \rho n}} 
=  \delta \cdot q^{\rho n} \;, 
\end{align*}
where the last inequality uses the inequalities (\ref{eq:no-leakage}) and (\ref{eq:cond-min-entropy}).
\end{proof}

Similar to the above, we can construct weak AMD codes with leakage-resilience from a weak AMD code without leakage-resilience. The proof (and even the parameter changes in the formulas) of the following lemma is similar to that of Lemma~\ref{lem:creatingleakage}, but we include it here for completeness.

\begin{lemma}\label{lem:creatingleakageweak}
For any $\delta > 0$, $0 < \rho < 1$, any regular coding scheme $\Enc:\F^k \mapsto \F^n$, $\Dec: \F^n \mapsto \F^k \cup \{ \bot \}$ that is a weak ($0, \delta$)-AMD code is also a weak ($\rho, q^{\rho n} \delta$)-AMD code.
\end{lemma}
\begin{proof}
Since $(\Enc, \Dec)$ is a $(0,\delta)$-AMD code, we have that for a uniform $M$ in $\F^k$, and any $\alpha \in \F^n$,
\begin{equation*}
\Pr(\Dec(\Enc(M) + \alpha) \notin \{M, \bot\}) \le \delta \;.
\end{equation*}
Define $\BAD(\alpha)$ to be the set of all $c$ in the support of $\Enc(M)$ such that $\Dec(c + \alpha) \notin \{M, \bot\}$. From the equation above, we have that
\begin{equation}
\label{eq:no-leakage-weak}
|\BAD(\alpha)| \le \delta \cdot q^k \;.
\end{equation} 
Now, consider a leakage variable $Z$ such that $\thinf(\Enc(M) | Z) \ge \hinf(\Enc(M)) - \rho n \log q$. Since $(\Enc,\Dec)$ is a regular coding scheme, we have that $\hinf(\Enc(M)) = \hinf(M) = k \log q$, and hence $\thinf(\Enc(M) | Z) \ge (k - \rho n) \log q$. Thus, using the definition of conditional min-entropy, we have that
\begin{equation}
\label{eq:cond-min-entropy-weak}
\sum_{z \in \cZ} \Pr(Z = z) \cdot \max_{c \in \F^n} \Pr(\Enc(M) = c \; | \; Z = z)  \le \frac{1}{q^{k - \rho n}} \;.
\end{equation}
We now bound the probability of incorrect decoding when the adversary computes the offset given $Z$. 
\begin{align*}
& \Pr[\Dec(\Enc(M) + \mathsf{A}(Z)) \notin \{M, \bot\} ] \\
=& \sum_{z \in \cZ}  \Big( \Pr[\Dec(\Enc(M) + \mathsf{A}(Z)) \notin \{M, \bot\} \; |\; Z=z ] 
 \cdot \Pr[Z=z] \Big) \\
=& \sum_{z \in \cZ} \Pr[\Enc(M) \in \BAD(\mathsf{A}(Z)) \; |\; Z=z ] \cdot \Pr[Z=z] \\
\le& \sum_{z \in \cZ}  \Big(  |\BAD(\mathsf{A}(Z)| \max_{c \in \F^n} \Pr(\Enc(M) = c \; | \; Z = z) 
 \cdot \Pr(Z = z) \Big) \\
\le& \delta \cdot q^k \cdot \frac{1}{q^{k - \rho n}} 
= \delta \cdot q^{\rho n} \;, 
\end{align*}
where the last inequality uses the inequalities (\ref{eq:no-leakage-weak}) and (\ref{eq:cond-min-entropy-weak}).
\end{proof}

We now give a construction of an efficient coding scheme without any leakage. This construction closely resembles the construction of AMD codes from~\cite{CDFPW08}.

\begin{theorem}\label{thm:strong-AMD}
For any positive integers $k < q-2, \sigma$, there exists an efficient regular coding scheme $\Enc:\F^k \times \F^\sigma \mapsto \F^{k + 2\sigma}$, $\Dec: \F^{k + 2\sigma} \mapsto \F^k \cup \{ \bot \}$ that is a strong ($0, \left(\frac{k+1}{q}\right)^\sigma$)-AMD code.
\end{theorem}
\begin{proof}
Let $f: \F^k \times \F \mapsto \F$ be defined as
\[
\forall m \in \F^k, \; a \in \F, \; f(m, a) := a^{k+2} + \sum_{i=1}^k m_i a^i \;,
\]
where $m = (m_1, \ldots, m_k)$ such that $m_i \in \F$ for $i \in [k]$.  
Then consider the coding scheme is defined as
\begin{align*}
& \forall m \in \F^k, x \in \F^\sigma, \;  \\
& \Enc(m, x) := (m, x, f(m, x_1), \ldots, f(m, x_\sigma))\;,
\end{align*}
The decoding function $\Dec$ on input $m' \in \F^k, x' \in \F^\sigma, y_1, \ldots, y_\sigma \in \F$ checks whether $y_i = f(m', x'_i)$ for $i \in [\sigma]$. If there exists an $i$ such that $y_i \neq f(m',x'_i)$, then $\Dec(m',x', y_1, \ldots, y_\sigma) = \bot$, else $\Dec(m',x', y_1, \ldots, y_\sigma) = m'$. 

Clearly the coding scheme is regular. We now proceed to show that the scheme is secure.

For any $m \in \F^k$, $\alpha \in \F^k$, $\beta \in \F^\sigma$, $\gamma \in \F^\sigma$, and a uniform $X\in \F^\sigma$  we need to bound 
\[
\Pr(\Dec(\Enc(m, X) + (\alpha, \beta, \gamma)) \notin \{m, \bot\}) \;.
\]
Notice that if $\alpha = 0$, then the above probability is $0$, since by definition, for any $m, \beta, \gamma$, $\Dec(\Enc(m, X) + (\alpha, \beta, \gamma))$ is either $m$ or $\bot$. Also, if $\alpha \neq 0$, then $\Dec(\Enc(m, X) + (\alpha, \beta, \gamma))$ is either $m + \alpha \neq m$, or $\bot$. Thus, it is sufficient to bound 
\[
\Pr(\Dec(\Enc(m, X) + (\alpha, \beta, \gamma)) \neq \bot) 
\]
for  any $m \in \F^k$, $\alpha \in \F^k \setminus \{0\}$, $\beta \in \F^\sigma$, $\gamma \in \F^\sigma$, and a uniform $X \in \F^\sigma$. Using the independence of $X_1, \ldots, X_\sigma$, we have that
\begin{align*}
 &\Pr(\Dec(\Enc(m, X) + (\alpha, \beta, \gamma)) \neq \bot)  \\
 =& \prod_{j = 1}^\sigma \Pr \Bigg( X_j^{k+2} + \sum_{i=1}^k m_i X_j^i + \gamma_j = (X_j+\beta_j)^{k+2}  
  + \sum_{i=1}^k (m_i+\alpha_i) (X_j + \beta_j)^i \Bigg)\\
=& \prod_{j = 1}^\sigma \Pr\Bigg(  X_j^{k+2} + \sum_{i=1}^k m_i X_j^i + \gamma_j - (X_j+\beta_j)^{k+2} 
  - \sum_{i=1}^k (m_i+\alpha_i) (X_j + \beta_j)^i = 0 \Bigg)\\
\le& \left( \frac{k+1}{q}\right)^\sigma \;, 
\end{align*}
since

\begin{align*}
P(X_j) &= X_j^{k+2} + \sum_{i=1}^k m_i X_j^i + \gamma_j - (X_j+\beta_j)^{k+2} \\
& - \sum_{i=1}^k (m_i+\alpha_i) (X_j + \beta_j)^i
\end{align*}

is a non-zero polynomial in $X_j$ of degree at most $k+1$. To see that the polynomial is non-zero, note that if $\beta_j \neq 0$, then the co-efficient of $X_j^{k+2}$ in $P(X_j)$ is zero, and that of $X_j^{k+1}$ is $(k+2) \beta_j \neq 0$. On the other hand, if $\beta_j = 0$, then let $t$ be the largest index such that $\alpha_t \neq 0$. Note that one such index exists since $\alpha \neq 0$. Then, the coefficients of $X_j^{k+2}, \ldots, X_j^{t+1}$ in $P(X_j)$ are $0$, and that of $X^t$ is $-\alpha_t \neq 0$. 

The desired result follows.
\end{proof}

The following corollary immediately follows from Lemma~\ref{lem:creatingleakage} and Theorem~\ref{thm:strong-AMD}.
\begin{corollary} \label{cor:strong-AMD}
For any positive integers $k < q-2, \sigma$, and $0 < \rho < \frac{1}{2}$, there exists an efficient regular coding scheme $\Enc:\F^k \times \F^\sigma \mapsto \F^{k + 2\sigma}$, $\Dec: \F^{k + 2\sigma} \mapsto \F^k \cup \{ \bot \}$ that is a strong ($\rho, q^{\rho (k + 2\sigma) - \sigma} (k+1)^\sigma$)-AMD code.
\end{corollary}
We remark that assuming $q \gg k$, as long as $\rho < \frac{\sigma}{k+2\sigma}$ the above construction is secure.

The last expression is equivalent to $2 \rho (k + 2\sigma) + k < 2 \sigma + k$.
This,
since the code rate for $\Enc$, $\kappa = \frac{k}{k + 2\sigma}$,
may be rewritten as
$2 \rho + \kappa < 1$. This formulation is the most interesting for us since it clearly states the postulated tradeoff between $\rho$ and $\kappa$.

Now we will construct weak AMD codes without any leakage. Notice that a construction with similar parameters was already shown in~\cite{LSW16}. We include the construction here for completeness. In Section~\ref{sec:lower-bound}, we show that the parameters obtained in this construction are optimal. 
\begin{theorem}\label{thm:weak-AMD}
Let $k$ be a positive integer, and let the characteristic of $\F$ be greater than $2$. Then there exists an efficient regular coding scheme $\Enc:\F^k \mapsto \F^{k + 1}$, $\Dec: \F^{k + 1} \mapsto \F^k \cup \{ \bot \}$ that is a weak ($0, \frac{1}{q}$)-AMD code.
\end{theorem}
\begin{proof}
Let $g: \F^k  \mapsto \F$ be defined as
\[
\forall m \in \F^k,  \; g(m) := \sum_{j=1}^k m_j^2 \;.
\]

where $m = (m_1, \ldots, m_k)$ such that $m_j \in \F$ for $j \in [k]$.  
Then consider the coding scheme is defined as
\[
\forall m \in \F^k, \; \Enc(m) := (m, g(m))\;,
\]
The decoding function $\Dec$ on input $m' \in \F^k, y \in \F$ checks whether $y = g(m')$. If $y \neq g_i(m')$, then $\Dec(m',  y) = \bot$, else $\Dec(m', y) = m'$. 

For any $\alpha \in \F^k$, $\beta \in \F$, and a uniform $M\in \F^k$  we need to bound 
\[
\Pr(\Dec(\Enc(M) + (\alpha, \beta)) \notin \{M, \bot\}) \;.
\]
Notice that if $\alpha = 0$, then the above probability is $0$, since by definition, for any $\beta$, $\Dec(\Enc(M) + (\alpha, \beta))$ is either $M$ or $\bot$. Also, if $\alpha \neq 0$, then $\Dec(\Enc(M) + (\alpha, \beta))$ is either $M + \alpha \neq M$, or $\bot$. Thus, it is sufficient to bound 
\[
p = \Pr(\Dec(\Enc(M) + (\alpha, \beta)) \neq \bot) 
\]
for  any $\alpha \in \F^k \setminus \{0\}$, $\beta \in \F$, and a uniform $M \in \F^k$. Without loss of generality, let $\alpha_i \neq 0$. Then,
\begin{eqnarray*}
p &=& \Pr\left( \sum_{j=1}^k \left((M_j + \alpha_j)^2 - M_j^2\right) = \beta\right) \\
&=& \Pr\left( \sum_{j=1}^k \left(2 \alpha_j M_j + \alpha_j^2 \right) = \beta\right) \\
&=& \Pr( 2 \alpha_i M_i = A) \;,
\end{eqnarray*}
where $A$ depends on $M_1, \ldots, M_{i-1}, M_{i+1}, \ldots, M_k, \alpha_1, \ldots, \alpha_k, \beta$, and is independent of $M_i$. Thus, using the independence of $M_i$ from $M_1, \ldots, M_{i-1}, M_{i+1}, \ldots, M_k$, we see that $p = \frac{1}{q}$. 
\end{proof}

The following corollary immediately follows from Lemma~\ref{lem:creatingleakageweak} and Theorem~\ref{thm:weak-AMD}.
\begin{corollary} \label{cor:weak-AMD}
For any positive integers $k$, and $0 < \rho < \frac{1}{k+1}$, there exists an efficient regular coding scheme $\Enc:\F^k \mapsto \F^{k + 1}$, $\Dec: \F^{k + 1} \mapsto \F^k \cup \{ \bot \}$ that is a weak ($\rho, q^{\rho (k + 1) - 1} $)-AMD code.
\end{corollary}
Thus, as long as $\rho < \frac{1}{k+1}$, or in other words, $ \rho + \kappa < 1$, and $\rho < \kappa$, the above construction is secure.

\section{Optimal Bounds on Leakage and Code Rate}
\label{sec:lower-bound}

In this section we show that the constructions in Section~\ref{sec:const} are asymptotically optimal. In fact, even if we allow that adversary is polynomial-time bounded (as long as the leakage is arbitrary), there still does not exist a construction of leakage resilient AMD codes that allow a better tradeoff between the rate of the code and the allowed leakage.

The following corollary is immediate from Lemma~\ref{lem:bonferroni}. 
\begin{corollary}
\label{cor:bonferroni}
Let $\F_{N}$ be a finite field, and let $A, B$ be uniform and independent in $\F_{N}$. Let $y_1, \ldots, y_t$ be some fixed distinct elements of $\F_{N}$.  Also, let $S$ be a subset of $\F_{N}$. Then the probability that there exists $i \in [t]$ such that $A y_i + B \in S$ is at least
\[
\frac{t \cdot |S|}{N} - \frac{t^2 \cdot |S|^2}{2 N^2} \;.
\]
\end{corollary}
\begin{proof}
Let $E$ denote the main event described in the statement and also (for all $i \in [t]$) let $E_i$ be the event that $Ay_i + B \in S$.
Then, of course we have $E = \cup_{i=1}^t E_i$ and $\Pr(E_i) = \frac{|S|}{N}$.
Now, all is set in such a way, that the final result is a conclusion of a direct use of the Lemma~\ref{lem:bonferroni}. The only missing part is to show that events $E_1, \ldots, E_t$ are pairwise independent:

To see this, note that for any $a, b \in \F_N$, and any $i, j \in [t]$ such that $i  \neq j$, we have that
\begin{eqnarray*}
& & \Pr(A y_i + B = a, Ay_j + B = b)\\
& &= \Pr\left(A= \frac{b-a}{y_j - y_i}, B= \frac{by_i-a y_j}{y_i - y_j}\right) \\
& &\stackrel{(\star)}{=} \Pr\left(A= \frac{b-a}{y_j - y_i}\right) \Pr\left(B= \frac{by_i-a y_j}{y_i - y_j}\right) \\
& &= \frac{1}{N} \cdot \frac{1}{N}\\
& &=   \Pr(A y_i + B = a) \cdot \Pr( Ay_j + B = b) \;.
\end{eqnarray*}

(The step ($\star$) follows from the independence of $A$ and $B$.)

The above calculations finish the proof of the pairwise independence of $E_1, \ldots, E_t$ which also ends the whole proof of the Corollary \ref{cor:bonferroni}.

\end{proof}

We are now ready to prove the optimality of our leakage-resilient AMD codes.

\begin{theorem}
\label{thm:impossibility-strong}
For any $\rho \in (0,1)$, there does not exist a computationally secure strong $(\rho, \frac{3}{16})$-AMD code $\Enc:\bits^k \times \bits^\sigma \mapsto \bits^n$, $\Dec: \bits^n \mapsto \bits^k \cup \{\bot\}$ with $2\rho + \frac{k}{n} \ge 1$.
\end{theorem}
\begin{proof}
Let $R$ be uniform in $\bits^\sigma$. Also, let $t = 2^{(n-k)/2}$. We divide the proof in two cases.

\paragraph*{Case 1} There exists a message $m$ such that the support of $\Enc(m,R)$ has cardinality at most $t$. Let $\cC(m)$ be the support of $\Enc(m,R)$. We define the set of good codewords $\cG \subseteq \cC(m)$ such that
\[
\cG := \left\{ c \in \cC(m) \: : \: \Pr(\Enc(m,R) = c) \ge \frac{1}{2t} 
\right\}\;.\]
The probability that $\Enc(m,R) \notin \cG$ is less than $\frac{1}{2t} \cdot t = \frac{1}{2}$. Thus, the probability that $C = \Enc(m, R) \in \cG$ is greater than $\frac{1}{2}$. 

Now, we interpret the domain of the randomness of the $\Enc$ function, i.e., $\bits^\sigma$ as a finite field of size $2^\sigma$, and let $A, B$ be uniformly random variables in $\bits^\sigma$ chosen by the adversary. Also, let $y_1, \ldots, y_{t}$ be fixed and distinct elements of $\bits^\sigma$. 

We define the random variable $Z$ to be the index $i \in [t]$ such\footnote{We note here, that index $i$ has $t$ values, thus by lemma \ref{lem:condmin}, leakage $Z$ fulfills our noisy leakage definition, and is in line with definition \ref{def:AMD}.} that $\Enc(m,Ay_i + B) = \Enc(m, R)$. If no such $i$ exists, then $Z$ is chosen to be an arbitrary index in  $[t]$. Furthermore, we fix a codeword $c^*$ such that $\Dec(c^*) \notin \{m, \bot\}$. 

Then the function $\mathsf{A}(Z)$ chosen by the adversary is defined as
\[
\mathsf{A}(Z) = -\Enc(m, Ay_Z + B) + c^* \;.
\]
Notice that $\Dec(\Enc(m,R) + \mathsf{A}(Z) )= \Dec(c^*)$ if $\Enc(m, R) \in \cG$ and there exists an $i \in [t]$ such that $\Enc(m,Ay_i + B) = \Enc(m, R)$. 
Conditioned on the event that $\Enc(m, R) \in \cG$ (which happens with probability greater than $1/2$), the number of $r \in \bits^\sigma$ such that $\Enc(m,r) = \Enc(m,R)$ is at least $\frac{2^\sigma}{2t}$. Thus, using Corollary~\ref{cor:bonferroni}, the probability that there exists an $i \in [t]$ such that $\Enc(m,Ay_i + B) = \Enc(m, R)$ is at least
\[
\frac{t}{2t} - \frac{t^2}{2 \cdot 4t^2} = 3/8\;. 
\]
Thus, $\Dec(\Enc(m,R) + \mathsf{A}(Z) )= \Dec(c^*) \notin \{m, \bot\}$ with probability $3/16$. 
\paragraph*{Case 2} For every message $m'$, the support of $\Enc(m',R)$ has cardinality greater than $t$. We fix a message $m \in \bits^k$, and the codeword $C = \Enc(m, R)$. 

Now, we interpret the code space, i.e., $\bits^n$ as a finite field of size $2^n$, and let $A, B$ be uniformly random variables in $\bits^n$ chosen by the adversary. Also, let $y_1, \ldots, y_{t}$ be fixed and distinct elements of $\bits^n$. We define the random variable $Z$ to be the index $i \in [t]$ such that $\Dec(C + Ay_i + B) \notin \{m, \bot\}$. If no such $i$ exists, then $Z$ is chosen to be an arbitrary index in  $[t]$. The function $\mathsf{A}(Z)$ chosen by the adversary is defined as $Ay_Z + B$. 

We now compute the probability that $\Dec(C+ \mathsf{A}(Z)) \notin \{m, \bot\}$. The number of strings $x \in \bits^n$ such that $\Dec(C + x) \notin \{m, \bot\}$ is greater than $(2^k-1)t$ since by assumption, for every message $m'$, the support of $\Enc(m',R)$ has cardinality greater than $t$. Thus, using Corollary~\ref{cor:bonferroni}, the probability that  there exists an index $i \in [t]$ such that $\Dec(C + Ay_i + B) \notin \{m, \bot\}$ is at least
\[
\frac{t^2 (2^k-1)}{2^n} - \frac{t^4 (2^k -1)^2}{2 \cdot 2^{2n}} =  \frac{1}{2} -\frac{1}{2^{2k+1}} \ge \frac{3}{8} \;,
\]
for $k \ge 1$. 

\end{proof}

\begin{remark}
We note that the adversary strategies considered in the proof of Theorem~\ref{thm:impossibility-strong} are reminiscent of the strategies used by Paterson and Stinson~\cite{PS16} to obtain bounds on the parameters of weak AMD codes. In particular, Case 1 is similar to the Guess strategy where guessing the codeword or leaking the randomness that generates the codeword is easy, and Case 2 is similar to the Random strategy where a randomly selected offset will be successful with high probability. 
\end{remark}

Theorem~\ref{thm:impossibility-strong} shows that if $2 \rho + \kappa \ge 1$, where the leakage rate is $\rho$ and the code rate is $\kappa$, then strong leakage-resilient AMD codes are impossible. This, combined with Corollary~\ref{cor:strong-AMD}, gives an (asymptotically)  tight characterization of code rate and leakage rate for which leakage-resilient strong AMD codes exist. 

Next, we show that there is no leakage-resilient (even computationally secure) weak AMD code from $k$-bit messages to $n$-bit codewords with leakage rate $\rho$ such that $\rho + \frac{k}{n} \ge 1$. 
\begin{theorem}
\label{thm:impossibility-weak}
For any $\rho \in (0,1)$, there does not exist a computationally secure weak $(\rho, \frac{3}{8})$-AMD code $\Enc:\bits^k \mapsto \bits^n$, $\Dec: \bits^n \mapsto \bits^k \cup \{\bot\}$ with $\rho + \frac{k}{n} \ge 1$ or $\rho \ge \frac{k}{n}$.
\end{theorem}
\begin{proof}
Let $M$ be chosen uniformly at random in $\bits^k$, and let $C = \Enc(M)$ be fixed. 
If $\rho \ge \frac{k}{n}$, then let $Z$ be $\Dec(C)$, and let $\mathsf{A}(Z)$ be $-\Enc(Z) + \Enc(m')$ for some $m' \in \bits^k \setminus\{ M\}$. Then, it is easy to see that $\Dec(C+ \mathsf{A}(Z)) = m' \notin \{M, \bot\}$ with probability $1$. This shows that $\rho < \frac{k}{n}$. 

Let $t = 2^{n-k}$. We interpret the code space, i.e., $\bits^n$ as a finite field of size $2^n$, and let $A, B$ be uniformly random variables in $\bits^n$ chosen by the adversary. Also, let $y_1, \ldots, y_{t}$ be fixed and distinct elements of $\bits^n$. We define the random variable $Z$ to be the index $i \in [t]$ such\footnote{We note that  index $i$ has $t$ values. By lemma \ref{lem:condmin}, leakage $Z$ fulfills the requirements set in definition \ref{def:weak-AMD}.} that $\Dec(C + Ay_i + B) \notin \{M, \bot\}$. If no such $i$ exists, then $Z$ is chosen to be an arbitrary index in  $[t]$. The function $\mathsf{A}(Z)$ chosen by the adversary is defined as $Ay_Z + B$. 

We now compute the probability that $\Dec(C+ \mathsf{A}(Z)) \notin \{M, \bot\}$. The number of strings $x \in \bits^n$ such that $\Dec(C + x) \notin \{m, \bot\}$ is $2^k-1$, one corresponding to each message in $\bits^k \setminus \{M\}$. Thus, using Corollary~\ref{cor:bonferroni}, the probability that  there exists an index $i \in [t]$ such that $\Dec(C + Ay_i + B) \notin \{M, \bot\}$ is at least
\[
\frac{t (2^k-1)}{2^n} - \frac{t^2 (2^k -1)^2}{2 \cdot 2^{2n}} =  \frac{1}{2}  - \frac{1}{2^{2k+1}} \ge \frac{3}{8} \;,
\]
for $k \ge 1$.
\end{proof}

Theorem~\ref{thm:impossibility-weak} shows that if $\rho + \kappa \ge 1$, where the leakage rate is $\rho$ and the code rate is $\kappa$, then weak leakage-resilient AMD codes are impossible. This, combined with Corollary~\ref{cor:weak-AMD}, gives an (asymptotically)  tight characterization of code rate and leakage rate for which leakage-resilient weak AMD codes exist. 

\section{Breaking the $\rho  < \frac{1}{2}$ barrier for AMD codes in the Ideal Cipher Model}
\label{breakbound}
In Section \ref{sec:lower-bound} we state an inequality $2\rho + \frac{k}{n} \le 1$ that must hold for all strong $(\rho, 3/16)$-AMD codes as introduced in Definition~\ref{def:AMD}. The definition assumes that the leakage variable $Z$ is arbitrary with the only constraint being that the entropy of the codeword conditioned on the knowledge of $Z$ is $1-\rho$ fraction of the original entropy. However, as we will see in this section, our result does not necessarily hold if we impose further conditions on the leakage variable $Z$.

We will work in the Ideal Cipher Model (abbr. ICM), which is equivalent to the Random Oracle Model, see \cite{CHKPST16}). As a reminder: ICM is a model with a public oracle (accessible fully to all parties) that gives access to a family $\{ f_i \}_{i\in I}$ of random (and independent) permutations. Any party may ask for $f_i(u)$ and for $f_i^{-1}(u)$ for any chosen $i$ and $u$. For our application we can even simplify the model and assume that we have only one random permutation $f$ with access to both forward and inverse queries to the oracle for $f$.

We will consider the following simple encoding: $\Enc_{br}(m, r) = (m, f(m,r))$, where $m \in \F^k$, $f: \F^k \times \F^\sigma \mapsto \F^k \times \F^{k+\sigma}$ is the oracle permutation from ICM described above and $r\in\F^\sigma$ is some (potentially huge in comparison to $m$) randomness. Obviously the function $\Enc_{br}$ is efficient and also there exists efficient decoding function that may efficiently (using $f^{-1}$ oracle) check if the codeword is correct.

Now we are ready for introducing the theorem about an interesting property of $\Enc_{br}$. Please note that this time we assume something about $Z$ from Definiton \ref{def:AMD}.
More specifically: We assume that $Z$ is computed by some Turing Machine (called leakage oracle) that is bounded by the number of queries to the ICM oracle.

\begin{theorem}
The above $\Enc_{br}$ is a strong $(\rho , q^{-k} + t / q^{\sigma - \rho n}))$-AMD code in the Ideal Cipher Model if the number of queries to ICM oracle (made together by both the adversary and the leakage oracle) is bounded by $t$. 
\end{theorem}
\begin{proof}
Let us assume that the adversary knows original $m$ and let us denote by $m'$ the new value of the message in the modified codeword. Also let us denote $x = f(m, r)$. Now the goal of the adversary is to construct $e$ such that $f^{-1}(x+e) = (m', r')$ for some $r'$.

During the execution, the adversary and the leakage oracle learned at most $t$ values $y_1, \ldots, y_t$ such that $f^{-1}(y_j) = (m', r')$ for some $r'$. This means that if $(x+e)$ is not equal to any of such $y_j$ for $j \in [t]$, then the result is simply uniform at random so the probability of success is exactly $q^{-k}$.

So, the only hope to make the probability of winning greater is to pick $e$ such that $x + e = y_j$ for some $1\leq j \leq t$. However since $\thinf(x|Z) \geq \sigma \log q  - \rho n \log q$ then this happens with probability less or equal to $\frac{t}{ 2^{\sigma \log q  - \rho n \log q} }= \frac{t} {q^{\sigma - \rho n}}$.

These two facts about two cases above together imply the statement.
\end{proof}

Thus, this construction gives an AMD code in the Ideal Cipher Model with $\rho \approx \frac{\sigma}{2k+\sigma}$, and $\kappa = \frac{k}{2k+\sigma}$. In this case, we can achieve leakage arbitrarily close to $1$ by having $k \ll \sigma$. 

\section{Conclusion}

Theorem~\ref{thm:impossibility-strong} and Corollary~\ref{cor:strong-AMD} show that strong leakage-resilient AMD codes exist if and only if $2 \rho + \kappa < 1$, where the leakage rate is $\rho$ and the code rate is $\kappa$.

Similarly, Theorem~\ref{thm:impossibility-weak} and Corollary~\ref{cor:weak-AMD} show that weak AMD codes exist if and only if $\rho + \kappa < 1$. 

Our results are asymptotically tight, but it would be interesting to prove that our bounds are tight when the message size is small. We leave this as an open question. 

\section{Acknowledgments}

The first named author would like to thank the Singapore Ministry of Education and the National Research Foundation for their supporting grants. The first and third authors were supported by the grant MOE2019- T2-1-145 Foundations of quantum-safe cryptography. The second named author was supported by the Polish National Science Centre (NCN) SONATA GRANT UMO-2014/13/D/ST6/03252 and by the TEAM/2016- 1/4 grant from the Foundation for Polish Science. The third named author would like to thank the European Research Council for the funding received under the European Unions’s Horizon 2020 research and innovation programme under grant agreement No 669255 (MPCPRO).

\bibliographystyle{fundam}
\bibliography{fuzzy2}

\begin{thebibliography}{10}
\providecommand{\url}[1]{\texttt{#1}}
\providecommand{\urlprefix}{URL }
\expandafter\ifx\csname urlstyle\endcsname\relax
  \providecommand{\doi}[1]{doi:\discretionary{}{}{}#1}\else
  \providecommand{\doi}{doi:\discretionary{}{}{}\begingroup
  \urlstyle{rm}\Url}\fi
\providecommand{\eprint}[2][]{\url{#2}}

\bibitem{CDFPW08}
Cramer R, Dodis Y, Fehr S, Padr{\'o} C, Wichs D.
\newblock Detection of Algebraic Manipulation with Applications to Robust
  Secret Sharing and Fuzzy Extractors.
\newblock In: Smart N (ed.), Advances in Cryptology -- EUROCRYPT 2008. Springer
  Berlin Heidelberg, Berlin, Heidelberg, 2008 pp. 471--488.

\bibitem{AS13}
Ahmadi H, Safavi{-}Naini R.
\newblock Detection of Algebraic Manipulation in the Presence of Leakage.
\newblock \emph{Information Theoretic Security - 7th International Conference,
  {ICITS} 2013, Singapore, November 28-30, 2013, Proceedings}, 2013.
\newblock pp. 238--258.

\bibitem{LSW16}
Lin F, Safavi{-}Naini R, Wang P.
\newblock Detecting Algebraic Manipulation in Leaky Storage Systems.
\newblock \emph{Information Theoretic Security - 9th International Conference,
  {ICITS} 2016, Tacoma, WA, USA, August 9-12, 2016, Revised Selected Papers},
  2016.
\newblock pp. 129--150.

\bibitem{Sha79}
Shamir A.
\newblock How to Share a Secret.
\newblock \emph{Communications of the {ACM}}, 1979.
\newblock \textbf{22}(11):612--613.

\bibitem{DORS06}
Dodis Y, Ostrovsky R, Reyzin L, Smith A.
\newblock Fuzzy Extractors: How to Generate Strong Keys from Biometrics and
  Other Noisy Data.
\newblock Technical Report 2003/235, Cryptology {ePrint} archive, {\tt
  http://eprint.iacr.org}, 2006.
\newblock Previous version appeared at {\em EUROCRYPT 2004}.

\bibitem{TW88}
Tompa M, Woll H.
\newblock How to Share a Secret with Cheaters.
\newblock \emph{Journal of Cryptology}, 1989.
\newblock \textbf{1.3}:133--138.

\bibitem{Boy04}
Boyen X.
\newblock Reusable Cryptographic Fuzzy Extractors.
\newblock In: 11th ACM Conference on Computer and Communication Security. ACM,
  2004 .

\bibitem{BDKOS05}
Boyen X, Dodis Y, Katz J, Ostrovsky R, Smith A.
\newblock Secure Remote Authentication Using Biometric Data.
\newblock In: Cramer R (ed.), Advances in Cryptology---EUROCRYPT 2005, volume
  3494 of \emph{LNCS}. Spring{\-}er-Ver{\-}lag, 2005 pp. 147--163.

\bibitem{Koc96}
Kocher P.
\newblock Timing Attacks on Implementations of {D}iffie-{H}ellman, {RSA},
  {DSS}, and Other Systems.
\newblock In: Koblitz N (ed.), Advances in Cryptology---CRYPTO~'96, volume 1109
  of \emph{LNCS}. Spring{\-}er-Ver{\-}lag, 1996 pp. 104--113.

\bibitem{KJJ99}
Kocher P, Jaffe J, Jun B.
\newblock Differential Power Analysis.
\newblock In: Wiener M (ed.), Advances in Cryptology---CRYPTO~'99, volume 1666
  of \emph{LNCS}. Spring{\-}er-Ver{\-}lag, 1999 pp. 388--397.

\bibitem{AARR02}
Agrawal D, Archambeault B, Rao JR, Rohatgi P.
\newblock The {EM} Side-Channel(s).
\newblock In: Jr BSK, \c{C}etin Kaya~Ko\c{c}, Paar C (eds.), CHES, volume 2523
  of \emph{Lecture Notes in Computer Science}. Springer.
\newblock ISBN 3-540-00409-2, 2002 pp. 29--45.

\bibitem{BM84}
Blakley GR, Meadows C.
\newblock Security of Ramp Schemes.
\newblock \emph{Crypto}, 1984.
\newblock \textbf{LNCS 196}:242--268.

\bibitem{IY06}
Iwamoto M, Yamamoto H.
\newblock Strongly secure ramp secret sharing schemes for general access
  structures.
\newblock \emph{Information Processing Letters}, 2006.
\newblock \textbf{97}(2):52--57.

\bibitem{JM96}
Jackson WA, Martin KM.
\newblock A combinatorial interpretation of ramp schemes.
\newblock \emph{Australasian Journal of Combinatorics}, 1996.
\newblock \textbf{14}:51--60.

\bibitem{DP08}
Dziembowski S, Pietrzak K.
\newblock Leakage-Resilient Cryptography.
\newblock In: 49th Symposium on Foundations of Computer Science. IEEE Computer
  Society, Philadelphia, PA, USA, 2008 pp. 293--302.

\bibitem{NS09}
Naor M, Segev G.
\newblock Public-Key Cryptosystems Resilient to Key Leakage.
\newblock In: Halevi S (ed.), Advances in Cryptology - CRYPTO~2009, volume 5677
  of \emph{LNCS}. Spring{\-}er-Ver{\-}lag, 2009 pp. 18--35.

\bibitem{DORS08}
Dodis Y, Ostrovsky R, Reyzin L, Smith A.
\newblock Fuzzy Extractors: How to Generate Strong Keys from Biometrics and
  Other Noisy Data.
\newblock \emph{SIAM Journal on Computing}, 2008.
\newblock \textbf{38}(1):97--139.

\bibitem{Bon36}
Bonferroni.
\newblock Teoria statistica delle classi e calcolo delle probabilita.
\newblock \emph{Libreria internazionale Seeber}, 1936.

\bibitem{PS16}
Paterson MB, Stinson DR.
\newblock Combinatorial characterizations of algebraic manipulation detection
  codes involving generalized difference families.
\newblock \emph{Discrete Mathematics}, 2016.
\newblock \textbf{339}(12):2891--2906.

\bibitem{CHKPST16}
Coron J, Holenstein T, Kunzler R, Patarin J, Seurin Y, Tessaro S.
\newblock How to Build an Ideal Cipher: The Indifferentiability of the Feistel
  Construction.
\newblock \emph{Journal of Cryptology}, 2016.
\newblock \textbf{29}(1):61--114.

\end{thebibliography}

\end{document}